\documentclass[runningheads,envcountsame]{llncs}
\usepackage[T1]{fontenc}
\usepackage{graphicx}
\usepackage{hyperref}
\usepackage{color}

\usepackage{amsmath,amssymb}

\usepackage{mathtools}
\usepackage{todonotes}
\usepackage{booktabs}
\newcommand{\VAR}[0]{\mathrm{VAR}}

\newcommand{\FO}[0]{\mathrm{FO}}
\newcommand{\MC}[0]{\mathrm{MC}}

\newcommand{\PMC}[0]{\mathrm{p\text-MC}}

\newcommand{\simple}[0]{\mathrm{simple\text-}}

\newcommand{\SSS}{\mathrm{SUBSET\text-SUM}}
\newcommand{\ALTSSS}[0]{\mathrm{ALT}_3\SSS}
\newcommand{\pALTSSS}[0]{\mathrm{p\text-ALT}_3\SSS}
\newcommand{\ALTTSSS}[0]{\mathrm{ALT}_\ell\SSS}
\newcommand{\pALTTSSS}[0]{\mathrm{p\text-ALT}_\ell\SSS}
\newcommand{\coALTTSSS}[0]{\mathrm{CO\text-ALT}_\ell\SSS}
\newcommand{\pcoALTTSSS}[0]{\mathrm{p\text-CO\text-ALT}_\ell\SSS}
\newcommand{\Fr}{\mathrm{Fr}}
\newcommand{\dom}{\mathrm{dom}}
\newcommand{\calA}{\mathcal{A}}
\newcommand{\leqfpt}{\leq^{\mathrm{fpt}}}
\newcommand{\A}[1]{\mathrm{A}[#1]}

\spnewtheorem{myclaim}{Claim}{\itshape}{\itshape}

\begin{document}
\title{A SUBSET-SUM Characterisation of the A-Hierarchy}
\author{Jan Gutleben \and
Arne Meier\orcidID{https://orcid.org/0000-0002-8061-5376}}
\authorrunning{J. Gutleben and A. Meier}
\institute{Institut für Theoretische Informatik, Leibniz Universität Hannover\\
Appelstrasse 9a, 30167 Hannover, Germany\\
\email{\{gutleben,meier\}@thi.uni-hannover.de.de}}
\maketitle              % typeset the header of the contribution
\begin{abstract}
The A-hierarchy is a parametric analogue of the polynomial hierarchy in the context of paramterised complexity theory. 
We give a new characterisation of the A-hierarchy in terms of a generalisation of the SUBSET-SUM problem.

\end{abstract}
\section{Introduction}
Classical worst-case complexity relates the difficulty of a problem in general to its input length.
In particular, when dealing with intrinsically hard problems, the input length is a very coarse measure and does not tell much about the ``structure'' of the problem. 
That is why parameterised complexity theory was introduced~\cite{DBLP:series/mcs/DowneyF99}. 
Here, one considers the complexity of a problem with respect to a \emph{parameter} and aims for a more fine-grained complexity analysis. 
Such parameterised problems can be formalised as tuples $(Q,\kappa)$ where $Q$ is a decision problem and $\kappa$ is a polynomial-time computable function, the \emph{parameterisation}.
This function maps inputs to often natural numbers that are seen as the parameter of the problem.
For instance, the parameter could be the size of a vertex cover in a graph or the number of variables in a formula. 
Depending on the chosen parameter, the runtime to solve the problem can now be formulated as a function of the parameter \emph{and} the input length.
Of course, it is important to study parameters that are relevant to practical scenarios; from the above, the number of variables in a formula is probably not such a good choice, as it is usually neither constant nor growing slowly with the input length.

A problem is said to be fixed-parameter tractable (FPT) if it can be solved in time $f(\kappa(x))\cdot n^{O(1)}$ for all inputs $x$ with $|x|=n$ for some computable function $f$. 
There exist different hierarchies above FPT, such as the W- and the A-hierarchy. 
Hardness for base classes of these hierarchies corresponds to intractability of the problem in the parameterised world, while membership in FPT is seen as tractability.
The W-hierarchy is defined in terms of weighted circuit satisfiability problems and \emph{weft}, which is the maximum number of large gates from an input to the input in a circuit~\cite{DBLP:journals/tcs/DowneyFR98}. 
The initial definition of the A-hierarchy was given via a short halting problem of alternating single-tape Turing machines~\cite{DBLP:journals/siamcomp/FlumG01}. 
Yet, further machine characterisations are known~\cite{DBLP:journals/tcs/ChenFG05}. 
The bridge to predicate logic through the model-checking problem for $\Sigma_\ell$-formulas (predicate formulas in prenex normal form with $\ell$ alternations) with bounded arity relations was established by Flum and Grohe~\cite{DBLP:journals/siamcomp/FlumG01}. 
The next step to unbounded arity relations is also due to Flum and Grohe~\cite{DBLP:journals/lmcs/FlumG05} (see, in the preliminaries section, Lemma~\ref{normalisation lemma}).

Interestingly to note, the A-hierarchy can be seen as a direct parametric analogue of the polynomial hierarchy. 
That is why this hierarchy is also of indepentent interest as it bridges the gap between classical complexity theory and parameterised complexity theory.

\paragraph{Contributions.} 
In this work, we will give a new characterisation of the A-hierarchy in terms of generalised SUBSET-SUM problems. 
As a start, we will show that a certain variant of the SUBSET-SUM problem is complete for the third level of the A-hierarchy. 
Furthermore, we will generalise this problem in a way to find complete ones for every level of the A-hierarchy. 
The benefit of this characterisation is that these SUBSET-SUM problem variants can be seen more natural than generic machine problems or model checking problems in predicate logic.

\paragraph{Organisation.}
At first, we will give a brief introduction to notions in predicate logic and parameterised complexity theory. 
Afterwards, we layout notions of alternating random access machines (ARAMs) and tail-nondeterminism. 
In the main part, we will give a thorough proof of $\A3$-completeness of a certain variant of the SUBSET-SUM problem. 
Finally, we will generalise this problem to arbitrary alternations and explain how the previous proof can be generalised to show completeness for every level of the A-hierarchy.
We finish with an summary and an outlook.

\section{Preliminaries}
We assume basic familiarity with the concepts of computational complexity theory~\cite{DBLP:books/daglib/0018514,komp}. 

\paragraph{Predicate Logic.}
We will give a brief introduction into first-order logic~\cite{DBLP:books/daglib/0018121}.
Here, let $\tau$ be a \emph{first-order vocabulary} consisting of function symbols and relation symbols with an equality symbol `$=$'. 
Denote by $\VAR$ a countably infinite set of \emph{first-order variables}. 
Terms over $\tau$ are defined as usual. The set of first-order logic ($\FO$) is defined as
\[
	\psi \Coloneqq
	t_1 =t_2\mid 
	R(t_1,\dots,t_k)\mid
	\lnot R(t_1,\dots,t_k)\mid
	\psi\land\psi\mid
	\psi\lor\psi\mid
	\exists x\psi\mid
	\forall x\psi,
\]
where $t_i$ are terms for $1\leq i\leq k$, $R$ is a $k$-ary relation symbol from $\sigma$, $k\in\mathbb N$, and $x\in\VAR$.
Let us denote by $\VAR(\psi)$ the \emph{set of variables} in a formula $\psi$ and by $\Fr(\psi)$ the \emph{set of free variables} in $\psi$.
Regarding semantics, we consider $\FO$-formulas in the context of \emph{$\tau$-structures}. 
These are pairs $\calA=(A,\tau^\calA)$, where $A$ is the \emph{domain} of $\calA$ (when clear from the context, we write $A$ instead of $\dom(\calA)$). 
Here, $\tau^\calA$ corresponds to an interpretation of the function and relational symbols in the usual way (for instance, $t^\calA\langle s\rangle=s(x)$ if $t=x\in\VAR$).
For tuples of terms $\mathbf t=(t_1,\dots,t_n)$ for $n\in\mathbb N$, we write $\mathbf t^\calA\langle s\rangle$ for $(t_1^\calA\langle s\rangle, \dots, t_n^\calA\langle s\rangle)$.
We say an intrepretation $\calA$ \emph{models} a formula $\psi$, $\calA\models\psi$ in symbols, if $\psi$ evaluates to true under~$\calA$.
The \emph{model checking} problem in first-order logic is defined as
\[
\MC(\Phi) \coloneqq \left\{\,\langle\mathcal{A},\varphi\rangle \mid \mathcal{A} \text{ is a structure}, \varphi \in \Phi \text{ and } \mathcal{A} \models \varphi \,\right\}.
\]
In the following, we define important formula classes.
A quantifier-free formula is in $\Sigma_0$ or also in $\Pi_0$. 
Let $t,n \in \mathbb{N}_0$. 
Then, we let for $\varphi \in \Pi_t$ and $\psi \in \Sigma_t$:
\begin{align*}
    &\exists x_1 \, \exists x_2 \, \cdots \exists x_n \, \varphi \in \Sigma_{t+1},\qquad \forall x_1 \, \forall x_2 \, \cdots \forall x_n \, \psi \in \Pi_{t+1} 
\end{align*}
If $\Phi$ is a class of formulas, then $\Phi[r]$ is the subclass of formulas in $\Phi$ that contain at most $r$-ary relations.
Furthermore, a $\Sigma_t$-formula is \emph{simple} if its quantifier-free part is a conjunction of atoms if $t$ is odd, and a disjunction of atoms if $t$ is even. 
Now, denote by $\simple\Sigma_t$ the class of all simple $\Sigma_t$-formulas. 
Additionally, denote by $\simple\Sigma_t^+$ the subclass of $\simple\Sigma_t$ formulas that are negation-free, while $\simple\Sigma_t^-$ is the subclass of $\simple\Sigma_t$ formulas where each atomic formula is negated.

\paragraph{Parameterised Complexity Theory.}
For an introduction to the field of parameterised complexity theory, we refer the reader to the textbook by Flum and Grohe~\cite{Para}, or the one of Downey and Fellows~\cite{DBLP:series/txcs/DowneyF13,DBLP:series/mcs/DowneyF99}.
%However, for reasons of self-containedness, we will provide a brief overview of the most important concepts in this section.

A \emph{parameterisation} w.r.t.\ an alphabet $\Sigma$ is a function $\kappa \colon \Sigma^* \rightarrow \mathbb{N}^+$ which is computable in polynomial time.
A \emph{parameterised problem} is a pair $(Q,\kappa)$, where $Q \subseteq \Sigma^*$ is a decision problem and $\kappa$ is a parameterisation.

\begin{definition}
    Let $\Sigma$, $\Gamma$ be alphabets, and $(A,\kappa)$ and $(B,\iota)$  be parameterised problems with $A \subseteq \Sigma^*$ and $B \subseteq \Gamma^*$. 
    Then a function $f \colon \Sigma^* \rightarrow \Gamma^* $ is called an \emph{FPT-reduction}, written $(A,\kappa) \leqfpt (B,\iota)$, if the following is true for all $x\in\Sigma^*$:
    \begin{itemize}
        \item $x \in A \Leftrightarrow f(x) \in B$
        \item There is a computable function $g$ and a polynomial $p$ such that $f$ can be computed in $g(\kappa(x)) \cdot p(x)$ steps.
        \item There exists a computable function $h\colon \mathbb{N}^+ \rightarrow \mathbb{N}^+$, such that $\iota(f(x)) \leq h(\kappa(x))$.
    \end{itemize}
\end{definition}
If $(Q,\kappa)$ is a parameterised problem, then we write $[(Q,\kappa)]^{\textrm{fpt}}$ for the class of all parameterised problems $(Q',\kappa')$ such that $(Q',\kappa') \leqfpt (Q,\kappa)$. 
Intuitively, this is known as the \emph{fpt-closure} of $(Q,\kappa)$.

The parameterised model checking problem is defined as follows:
\[
    \PMC(\Sigma_t) \coloneqq (\MC(\Sigma_t), \langle \mathcal{A},\varphi \rangle \mapsto |\varphi|).
\]
Finally, we are ready to define the A-hierarchy. 
For all $t \in \mathbb{N}^+$ let
    %\[
    $A[t] \coloneqq [\PMC(\Sigma_t)]^{\textrm{fpt}}$.
    %\]
The union of these classes, $\bigcup_{t \in \mathbb{N}} A[t]$, is called the \emph{A-hierarchy}.

The following lemma states that for the parameterised model checking problem for $\Sigma_t$ formulas one can assume the normal form of positive simple $\Sigma_t$-formulas.

\begin{lemma}[{\cite[L.~8.10]{Para}}] \label{normalisation lemma}
    For all $t \in \mathbb{N}^+$, 
    %\[
    $\PMC(\Sigma_t) \leqfpt \PMC(\simple\Sigma_t^+[2])$.%
    %\]
\end{lemma}

Notice that we will use a strengthened version of the previous lemma. 
Here, only binary relations appear in the formula. 
This is achieved via simulating the unary relations by binary relations with a dummy variable.

\paragraph{ARAMs and Tail-Nondeterminism.}
In the following, we introduce the necessary notions to characterise the class $\mathrm{A}[t]$ in terms of ARAMs. 
ARAMs are a generalisation of RAMs that allow for nondeterministic behaviour.
First, we start with classical RAMs and follow the notion of Flum and Grohe~ \cite[pp.\ 457]{Para}.
A \emph{Random Access Machine} (RAM) consists of countably infinite many registers $R_0, R_1 \dots$, a finite sequence of instructions $I_1, \dots I_n$, and a program counter that contains a number from $\mathbb{N}^+$.
Registers may contain numbers from $\mathbb{N}_0$. 
The register $R_0$ is called the \emph{accumulator}. 
We allow arithmetic instructions: ADD, SUB (cut off at $0$), or DIV2 (rounded off) with the usual semantics (notice that arithmetic instructions are executed in $O(1)$ time). 
The first registers serve for inputs ($x=x_1\dots x_n$ is in $R_1,\dots,R_n$) and outputs. 
Acceptance/rejection for decision problems is then via storing $1/0$ in $R_0$. 
The runtime of algorithms on such machines is then measured with respect to length of input and the number of instructions carried out (no matter how large the involved numbers are).

\begin{definition}
    An \emph{alternating Random Access Machine} (ARAM) is a RAM with two additional GUESS instructions: EXISTS and FORALL.
\end{definition}
An ARAM \emph{accepts} an input if for every FORALL instruction the machine accepts on every possible number equal or less than the number in the \emph{accumulator}, and for every EXISTS instruction there is a number on which the machine accepts.
We say that an ARAM is $t$-alternating if there are at most $t$ alternations between EXISTS and FORALL instructions. 

\begin{definition}\label{ARAM besch}
    For a parameterisation $\kappa \colon \Sigma^* \rightarrow \mathbb{N}^+$, an ARAM is \emph{$\kappa$-bounded} if there are computable functions $f,g$ and a polynomial $p$ such that for the ARAM on all inputs $x \in \Sigma^*$ the following is true.
    \begin{itemize}
        \item The program needs at most $f(\kappa(x))\cdot p(|x|)$ steps, of which $g(\kappa(x))$ are nondeterministic.
        \item The program uses at most the first $f(\kappa(x))\cdot p(|x|)$ registers.
        \item The program only uses numbers less than or equal to $f(\kappa(x))\cdot p(|x|)$.
    \end{itemize}
\end{definition}
Furthermore, there is a notion of tail-nondeterminism for ARAMs.

\begin{definition}\label{ARAM tail}
    Let $\kappa$ be a parameterisation. 
    A $\kappa$-bounded ARAM program is called \emph{tail-nondeterministic} if there is a computable function $h$ such that the nondeterministic steps of the program are always in the last $h(\kappa(x))$ steps.
\end{definition}

Now, we are ready to state a characterisation of the classes of the $A$-hierarchy in terms of ARAMs.
\begin{theorem}[{\cite[Thm.~8.8]{komp}}]\label{ARAm = in A}
    Let $t \in \mathbb{N}^+$ and $(Q,\kappa)$ be a parameterised problem. 
    Then $(Q,\kappa) \in A[t]$ if and only if there is a tail-nondeterministic $\kappa$-bounded $t$-alternating ARAM that decides $(Q,\kappa)$.
\end{theorem}

\section{Generalised SUBSET-SUM and the A-Hierarchy}
In the following, we will consider a variant of the well-known SUBSET-SUM problem. 
Classically, this problem is known to be NP-complete~\cite{DBLP:books/fm/GareyJ79}. 
As a first step, we will give a definition of a variant of that classical problem that is complete for the third level of the A-hierarchy. 
Afterwards, we will generalise this problem to show completeness for every level of the A-hierarchy.
The problem $\ALTSSS$ is defined as follows:
\[
\left\{ \langle A_1,k,A_2,l,A_3,m,t \rangle\;\middle |\; \parbox{7cm}{$A_1, A_2, A_3\subseteq\mathbb N_0$, $k,l,m,t \in \mathbb{N}_0$, $\exists A'_1  \subseteq A_1$ with $|A_1'| = k$, $\forall A_2' \subseteq A_2$ with $|A_2'| = l$,  $\exists A_3' \subseteq A_3$ with $|A_3'| = m$, such that $\sum_{a \in A_1'} a + \sum_{a \in A_2'} a + \sum_{a \in A_3'} a = t  $ } \right\},
\]
\begin{example}
    Consider the following instance of $\ALTSSS$: \[\langle\{0,3\},1,\{1,2\},1,\{2,3\},1,7\rangle\]
    If this instance is in $\ALTSSS$, than there must exist a subset of $\{0,3\}$ with magnitude 1 such that for all subsets with of $\{1,2\}$ with magnitude 1 there exists a subset of $\{2,3\}$ such that all the chosen subsets sum to 7. If $\{0\}$ is chosen out of $\{0,3\}$, than every possible sum is smaller than 7. That concludes, that if this is a \emph{yes}-Instance, $\{3\}$ must be chosen. Now every subset with magnitude 1 of $\{1,2\}$ must be investigated. If $\{1\}$ is chosen, out of the last set $\{3\}$ can be chosen. The sum is in this case $3+1+3 = 7$. The only other subset with correct size is $\{2\}$. In this case $\{2\}$ can be chosen out of the last set. Since $3+2+2=7$, the chosen sets sum up correctly. We can conclude: 
    \[\langle\{0,3\},1,\{1,2\},1,\{2,3\},1,7\rangle \in \ALTSSS\]
\end{example}

The parameterised version of this problem then is $\pALTSSS$ and is defined as 
\[
    (\ALTSSS, \langle A_1,k,A_2,l,A_3,m,t \rangle \mapsto k+l+m). \]

We will show that $\pALTSSS$ is complete for the third level of the A-hierarchy.
\begin{theorem}\label{sss vst}
    $\pALTSSS$ is $\A3$-complete.
\end{theorem}

We split the proof of the result into the following two lemmas.

\begin{lemma}\label{SSS3-in}
    $\pALTSSS$ is in $\A3$.
\end{lemma}
\begin{proof}
We give an ARAM program deciding membership. 
It uses $k+l+m$ registers for the numbers in the sets $A_1,A_2,A_3$. 
After these registers, the following ones store the values of $k,l,m,|A_1|,|A_2|,|A_3|$. 
Then, we need a register \texttt{sum} for the respective sum of the subsets. 
After that register, we use separate registers for the natural numbers in the three sets.
The nondeterministic instructions FORALL/EXISTS allow us to guess the required subsets by guessing indices after loading the ``offset'' $|A_i|$ from the corresponding register into $R_0$. 
These numbers are added up one by one into the \texttt{sum} register and compared to~$t$.

Regarding the runtime of the ARAM program, after preparing the registers in $O(n)$ steps, guessing of the subsets in $k+l+m$ steps as well as adding up (again $k+l+m$ steps) and comparing with $t$ can be overall done in $O(k+l+m)$ steps.
Clearly, the program is $3$-alternating and $(k+l+m)$-bounded. 
Hence, by Theorem~\ref{ARAm = in A}, we have that $\pALTSSS \in \A3$.\hfill$\Box$
\end{proof}

The idea of the following lemma is to thoroughly construct numbers over a particular basis such that satisfaction of the model-checking instance formula corresponds to summing up numbers reaching a target sum. 
In the course of defining these numbers, we have to be careful in two ways. 
First, we need to ensure that no overflow can occur. 
Second, we need to guarantee that, depending on the quantifiers in the formula, there are always choices for the ``universal'' set.
It is a good idea to consult Fig.~\ref{fig-numbers} first to get a vague idea of the numbers used.

\begin{lemma}\label{SSS3-hard}
    $\pALTSSS$ is $\A3$-hard.
\end{lemma}

\begin{proof}
We will show a reduction from $\PMC(\simple\Sigma_3[2])$.
For that let $\langle \mathcal{A},\varphi\rangle$ be an instance of $\PMC(\simple\Sigma_3[2])$.
Hence, $\mathcal{A} = (A, \tau)$, where $\tau$ only contains binary relations. 
Furthermore, for $x_{i,j} \in \{x_1,\dots, x_{k+l+m}\}$, we have that
\[
    \varphi = \exists x_1 \dots \exists x_k \forall x_{k+1} \dots \forall x_{k+l} \exists x_{k+l+1} \dots \exists x_{k+l+m} \bigwedge_{i=1}^n \lambda_i(x_{i,1},x_{i,2}),
\]
with $\lambda_i$ is an atom $\lambda_i = R(x_{i,1},x_{i,2})$ for some relation $R \in \tau$.
Without loss of generality, assume some fixed bijection $I\colon A \rightarrow [1,|A|]$ which refers to the index of an element in $A$.

Note that all numbers defined in the following are in some base $D$. 
This base will be later chosen large enough such that no overflows can occur.
Almost all numbers will be of the form $L_1\dots L_n \; B_1\dots B_{k+l+m}$ and the rightmost bit is the least significant bit.  
You can see a summary of the numbers used in the proof in Figure~\ref{fig-numbers}.
They will start with $n$ digits, abbreviated by $L_1,\dots,L_n$ (the only exception will be particularly defined numbers at the end, they will start with a larger block of other digits). 
These digits will yield information on the $\lambda_i$'s. 
Succeeding these digits, there will be $k+l+m$ blocks $B_j$ of each $k+l+m+1$ digits. 
These blocks will contain information on the variables $x_j$. 
Alltogether, these numbers have $n+(k+l+m+1)\cdot(k+l+m)$ digits.

\begin{figure}[t]
    \resizebox{\textwidth}{!}{\(\displaystyle%
    \begin{array}{r@{}c*{7}{@{\,}c}*{2}{@{\;}c}}
     & L_1&\cdots&L_i&\cdots&L_n&B_1&\cdots B_{j'}\cdots& B_j&\cdots&B_{k+l+m}\\
     &&&&&&&B^j_{j'}&\overbrace{B_j^1\cdots B_j^j\cdots B_j^{j'}\cdots B_j^{k+l+m+1}}&&\\
\text{VAR}(a,x_j)=&
    0&&\cdots&&0&
    0&\cdots&
    0\cdots I(a)\cdots0\cdots0 &\cdots&0\\
\text{ATOM}(a,b,x_j,x_{j'})= &0&\cdots&1&\cdots& 0& 0 &I(b)&0\cdots0\cdots I(a)\cdots0&\cdots&0\\
\text{NORM}(a,x_j)= & 0&&\cdots&&0&
0&\cdots0&
    (l+1)\cdot|A|\cdots (l+1)\cdot|A|1 &0\cdots&0\\
\text{FIX}(x_j)=&1&&\cdots&&1&
0&\cdots0&
    (l+1)\cdot|A|\cdots (l+1)\cdot|A|1 &0\cdots&0\\
\text{FIX}(x_j,d)=&
    0&&\cdots&&0&
    0&\cdots0&
    (l+1)\cdot|A|\cdots (l+1)\cdot|A|-d\cdots1&0\cdots&0 \\\bottomrule
\end{array}\)}

\resizebox{\textwidth}{!}{\(\displaystyle
\begin{array}{r@{\;}*{7}{c}}
    &1& \cdots & r &\cdots & s & L_1\cdots L_n & B_1\cdots B_{k+l+m}\\
    \text{WAIT}(r)=& 0&\cdots&1&\cdots&0&0\cdots0&0\cdots0\\
    \text{NOWAIT}=& 1&\cdots &1&\cdots&1&0\cdots0&0\cdots0\\
    t'=&1&\dots&1&\cdots&1&1\cdots1&(l+1) \cdot |A| \dots (l+1)|A|1\cdots (l+1) \cdot |A| \dots (l+1)|A|1
\end{array}
\)}

\caption{Overview of used numbers in the proof of Theorem~\ref{sss vst}. Here, $a,b\in A, 1\leq i\leq n$, $(b,c)\models\lambda_i(x_j,x_{j'})$, $d\in[0,l\cdot|A|]$, $1\leq j'< j\leq k+l+m$, $1\leq r\leq s$. Just for presentation reasons, we assumed $j'<j$. As part of the proof, we consider the possibility that $j<j'$ is also possible.}\label{fig-numbers}
\end{figure}

We will now define the sets $A_1,A_2,A_3$. 
We need to encode that a variable is assigned a value from the universe. 
For that purpose, we define for all $a \in A$ and $x_j \in \{x_1,\dots,x_{k+l+m}\}$ the number $\mathrm{VAR}(a,x_j) \coloneqq B_1 B_2 \dots B_{k+l+m}$.
Here, each block $B_j$ has the value $I(a)$ at the $j$-th position, all other positions are $0$.
These numbers are then distributed as follows:
 \begin{itemize}
     \item $j \in [1,k] \Rightarrow \mathrm{VAR}(a, x_j) \in \mathrm{VAR}_1 $
     \item $j \in [k+1,k+l] \Rightarrow \mathrm{VAR}(a, x_j) \in \mathrm{VAR}_2$
     \item $j\in [k+l+1,k+l+m] \Rightarrow \mathrm{VAR}(a, x_j) \in  \mathrm{VAR}_3$
 \end{itemize}
We need to ensure that $k,l,m$ elements from the respective VAR-block are chosen.
The next claim directly follows by the definition of the VAR-numbers.
\begin{myclaim}\label{lemma var}
    Let $\mathcal{J}$ be an assignment of $x_1, \dots, x_{k+l+m}$ and $a_j \coloneqq \mathcal{J}(x_j)$. Then we have that

    \[
        \sum_{j=1}^{k+l+m} \mathrm{VAR}(a_j, x_j) = B_1 B_2 \dots B_{k+l+m},
    \]
    where in $B_j$ at position $j$ the value is $I(a_j)$ and otherwise there are only zeros.
\end{myclaim}
Next, we will encode the atomic formulas.
For all $a,b \in A$ and $x_j,x_{j'} \in \{x_1,\dots,$ $x_{k+l+m}\}$ with $j < j'$, we define the number 
\[\mathrm{ATOM}(a, b, x_j, x_{j'}) \coloneqq L_1 \dots L_n B_1 \dots B_{k+l+m}.\]
For $i \in [1,n]$, $L_i$ is $1$ if and only if $(a,b) \models \lambda_i$ and $x_j,x_{j'}$ are the variables in $\lambda_i$; otherwise $L_i=0$ for all $i$.
All $B$-blocks are zero except for the $j$-th and $j'$-th.
Block $j$ has the value $I(a)$ at position $j'$, block $j'$ has the value $I(b)$ at position $j$ and otherwise zeros. 
All these numbers are collected in the set ATOM.
The following claim is about the sum corresponding to an assignment that fits to the ATOM- and VAR-numbers.
\begin{myclaim}\label{lemma atom}
Let $\mathcal{J}$ be an assignment of $x_1\dots x_{k+l+m}$. 
Then $\mathcal{J}$ satisfies the quantifier-free part of $\varphi$ if and only if
\[
\sum_{j = 1}^{k+l+m} \left(\mathrm{VAR}(a_j, x_j) + \sum_{j' = j+1}^{k+l+m} \mathrm{ATOM} (a_j, a_{j'}, x_j, x_{j'}) \right) = \underbrace{1\dots 1}_n\; B_1\dots B_{k+l+m},\]
with $B_j$ for all $j \in [1,k+l+m]$:
\[
B_j \coloneqq \overbrace{I(\mathcal{J}(x_j))\dots I(\mathcal{J}(x_j))}^{k+l+m}0.
\]
\end{myclaim}
\begin{proof}[of Claim~\ref{lemma atom}]
    The number is composed as follows. 
    At the beginning, it has exactly $n$ ones if $\mathcal{J}$ satisfies the quantifier-free part of $\varphi$. 
    The reason for this is that the only summand that can add a $1$ at position $L_i$ is ATOM$(a_j, a_{j'}, x_j, x_{j'})$ if $x_j$ and $x_{j'}$ are the variables in $\lambda_i$.
    There is a $1$ added if and only if $(a_j,a_{j'}) \models \lambda_i$.     
    For $x_j$, at each position $j'\neq j$ of the $B_j$ block, the value $I(\mathcal{J}(x_j))$ was added by adding the values of ATOM$(a_j, a_{j'}, x_j, x_{j'})$. 
    According to Claim~\ref{lemma var}, the value $I(\mathcal{J}(x_j))$ is already added by the number VAR$(a_j, x_j)$. 
    The $(k+l+m+1)$-th position of each block is not increased by any number, so it remains at $0$.\hfill$\blacksquare$
\end{proof}

As it is not clear in advance which variables are assigned to which elements of $A$, we need to ensure that the sum $t$ can always be reached. 
For that purpose, we introduce the numbers $\mathrm{NORM}(a, x_j) \coloneqq B_1 \dots B_{k+l+m}$. 
For all $j \in [1,k+l+m]$, all $B$-blocks except $B_j$ consist of zeros.
The block $B_j$ has the value $(l+1) \cdot |A| - I(a)$ except for the last position, where a $1$ is placed. 
The set of these numbers is called NORM.
Now, we define the target number $t$ as
\[
t \coloneqq \underbrace{1 \dots 1}_n\;B_1 \dots B_{k+l+m}.\]
For all $j \in [1,k+l+m]$ we have that
\[
B_j \coloneqq \underbrace{(l+1) \cdot |A| \dots\; (l+1) \cdot |A|}_{k+l+m} 1.
\]
The first $n$ digits ensure that every atom  satisfied. 
The remaining $(k+l+m) \cdot (k+l+m+1)$ digits are blocks with the value $(l+1) \cdot |A|$ followed by a $1$. 
The $1$ ensures that no more than one NORM-number is chosen for the same variable.
The next claim states that the sum $t$ can be reached by choosing numbers from VAR, ATOM, and NORM. 
It follows by construction of the numbers.
\begin{myclaim}\label{lemma norm}
    Let $\mathcal{J}$ be an assignment of $x_1,\dots, x_{k+l+m}$. 
    Then $\mathcal{J}$ satisfies the quanti\-fier-free part of $\varphi$ if and only if 
    \[
        \sum_{j = 1}^{k+l+m} \left(\mathrm{VAR} (a_j, x_j) + \mathrm{NORM}(a_j, x_j) + \sum_{j' = j+1}^{k+l+m} \mathrm{ATOM}(a_j, a_{j'}, x_j, x_{j'}) \right) = t.
    \]
\end{myclaim}

The following claim states that it is impossible to reach the sum $t$ by choosing subsets that are not according to Claim~\ref{lemma norm}.

\begin{myclaim}\label{lemma kaputt}
    Let $A'_1, A'_2, A_3'$ be some chosen subsets. 
    Exists a number $x_j$ for which neither VAR$(a, x_j)$ nor NORM$(a, x_j)$ has been chosen then it is impossible to reach the sum $t$, if only numbers from VAR$_1$, VAR$_2$, VAR$_3$, ATOM or NORM have been chosen.
\end{myclaim}
\begin{proof}[of Claim~\ref{lemma kaputt}]
    Consider the overall sum.
    If neither the number VAR$(a, x_j)$ nor NORM$(a, x_j)$ has been chosen, then a $0$ would be at position $j$ of the $B_j$-block. 
    This would in turn require two NORM-numbers to be added to reach $(l+1) \cdot |A|$. 
    This is because the largest number that can be added by a NORM-number is $(l+1) \cdot |A| - 1$. 
    This would produce a $2$ at the last position of the $B_j$-block.
    If no NORM$(a, x_j)$ has been chosen, at the last position of the $B_j$-block would be a $0$.
    Using only numbers from ATOM cannot solve this problem either, as they have at position $j$ a $0$. \hfill$\blacksquare$
\end{proof}

As $A_1$ and $A_3$ are existentially quantified, it is correct that subsets that do not correspond to a valid assignment cannot yield to the sum that is $t$. 
However, for $A_2$, all possible subsets must be able to reach the sum $t$ for a valid assignment.
If a variable is assigned a value twice in $A_2$ (which is an invalid assignment but a correct subset), we add numbers to $A_3$ that allow to still reach the sum $t$. 
If one or more variables are not assigned a value in $A_2$, there must be a variable, that has not been assigned, as from VAR$_2$ exactly $l$ numbers have to be chosen. 
We add numbers that set $L_1\dots L_n$ to all ones in that case. 
For all $j \in [k+1,k+l]$, we define 
\[
\mathrm{FIX}(x_j) \coloneqq \underbrace{1 \dots 1}_n\; B_1 \dots B_{k+l+m}.
\] 
All $B$-blocks are zero except for block $B_j$. 
This block has the value $(l+1) \cdot |A|$ at every position except a 1 at the last.

Now we need a possibility to set every digit in every $B$-block to $(l+1) \cdot |A|$. 
For that purpose, for all $d \in [0,l \cdot |A|]$ and $j \in [1,k+l+m]$ we add the numbers
\[
\mathrm{FIX}(x_j, d) \coloneqq B_1 \dots B_{k+l+m},
\] 
where all $B$-blocks consist of zeros except for block $B_j$, which has the value $(l+1) \cdot |A|$ at every position except a 1 at the last and at $j$ the value $(l+1) \cdot |A| - d$.
The set of these numbers is called FIX. 

\begin{myclaim}\label{lemma fix}
    Let $A_2' \subseteq $VAR$_2$ and assume that for an $x_{j}$ no VAR$(a, x_j)$ is in $A_2'$. 
    For every $j' \neq j$ there exists a $d_{j'}$, such that
    \begin{align*}
    &\sum_{j'=1}^k \mathrm{VAR}(a_{j'}, x_{j'}) + \sum_{a \in A_2'} a + \mathrm{FIX}(x_{j}) + \sum_{j'=1, j \neq j'}^{k+l+m} \mathrm{FIX}(x_{j'}, d_{j'}) = t.
    \end{align*}
\end{myclaim}
\begin{proof}[of Claim~\ref{lemma fix}]
    In this sum, every $L_i=1$ because of FIX$(x_{j'})$.
    Consider the $B_{j'}$-block for $j' \in [1,k+l+m]$. 
    If $j'=j$, by $A_2'$ no number has been added to $B_{j'}$ and the addition of FIX$(x_{j'})$ ensures that in the block the correct number is written. 
    If $j' \neq j$, between $0$ and $l$ numbers of the form VAR$(a, x_j)$ have been added because of $A_2'$. 
    This means that at position $j$ of the block a number from $[0,l \cdot |A|]$ is written. 
    Choose $d_j'$ for that value. 
    Then, the block has the desired form.\hfill$\blacksquare$
\end{proof}

We let $A_1 = \text{VAR}_1$, $A_2 = \text{VAR}_2$, and $A_3 = \text{VAR}_3 \cup \text{ATOM} \cup \text{NORM} \cup \text{FIX}$.
In order to choose subsets of the correct size, we need to choose $k$ numbers from VAR$_1$ and $l$ numbers from VAR$_2$. 
For the numbers from $A_3$ we need to make a case distinction:
\begin{enumerate}
    \item If a valid assignment is chosen from $A_2$, then according to Claim~\ref{lemma norm} the choice for $A_3$ is as follows:
    First, $m$ numbers from VAR are chosen. 
    Then, for every pair of variables, a number from ATOM is chosen. 
    This corresponds to $\frac{(k+l+m) \cdot (k+l+m-1)}{2}$ numbers. 
    Finally, $k+l+m$ numbers from NORM are chosen.
    \item If an invalid assignment is chosen from $A_2$, then according to Claim~\ref{lemma fix} the choice for $A_3$ also includes the FIX numbers. 
    Note that in each case the same number of variables must be chosen to reach $t$. 
    For that reason, let the difference between this cases be 
    \[
        s = (m+\frac{(k+l+m) \cdot (k+l+m-1)}{2} + k +l +m) - (k+l+m).
    \] 
    Observe that $s$ is greater than $0$ for all $k,l,m \geq 0$. 
    In that case, we need some further auxiliary numbers that can be chosen in the second case, to reach the same number of variables. 
    For all $i \in [0,s]$
    \begin{align*}
        \mathrm{WAIT}(i) \coloneqq 1\; \underbrace{0\dots0\; 0\dots0\; 0 \dots 0}_{\mathclap{i+n+(k+l+m)\cdot (k+l+m+1)}}\quad\mathrm{NOWAIT} \coloneqq \overbrace{1\dots 1}^{s+1}\; \underbrace{0\dots0\; 0\dots0\; 0 \dots 0}_{\mathclap{n+(k+l+m)\cdot (k+l+m+1)}}
    \end{align*}
    Observe that the sum of the WAIT$(i)$ equals NOWAIT. 
    These are the WAIT-numbers and are added to $A_3$. 
    We need one last adujstment of $t$ to 
    \[t'\coloneqq \overbrace{1\dots 1}^{s+1} t.\]
\end{enumerate}

\paragraph{The size of the basis.}
The base $D$ is intended to prevent overflows. 
We will now present out a worst-case analysis for a digit to show how an overflow can be prevented. 
In the proof, we choose $2\cdot(k+l+m) + \frac{(k+l+m) \cdot (k+l+m-1)}{2}+1$ numbers.
The highest value for one digit is $(l+1) \cdot |A|$.  
This allows us to choose 
\[D = \left(2\cdot(k+l+m) + \frac{(k+l+m) \cdot (k+l+m-1)}{2} + 1\right)\cdot (l+1) \cdot |A| +1.\]

\paragraph{The reduction function.}
We are ready to define the reduction function $f$. 
Let $f(\langle \mathcal{A},\varphi \rangle)$ be defined as
\[ \left\langle A_1,k,A_2,l,\tilde{A}_3,m+ \frac{(k+l+m) \cdot (k+l+m-1)}{2} + k+l+m+1,t'\right\rangle,\]
where $\tilde A_3 = A_3 \cup \mathrm{WAIT}$. We will now show, that $f$ is a correct reduction function.
Every digit is computable in O($|\langle \mathcal{A},\varphi\rangle|$) time. There are O($|\varphi|^2)$ digits in a number so every number is computable in O($|\langle \mathcal{A},\varphi\rangle|^3)$. There are O($|\langle \mathcal{A},\varphi\rangle|^4$) numbers computed, so $f$ is computable in O($|\langle \mathcal{A},\varphi\rangle|^7$).
Let $g$ be the computable function $g(x) = x^2 + 2\cdot x + 1$. The new parameter is bounded by $g$:
\[
    m + \frac{(k+l+m) \cdot (k+l+m-1)}{2} + k+l+m+1 \leq  |\varphi|^2 + 2 \cdot|\varphi| + 1 = g(|\varphi|).
\]
Finally, we turn towards the correctness property.
\begin{myclaim}\label{correcntessclaim}
It is true that \[
\langle \mathcal{A}, \varphi \rangle \in \mathrm{MC}(\simple\Sigma_3[2]) \Longleftrightarrow f(\langle \mathcal{A}, \varphi \rangle) \in  \pALTSSS.
\]    
\end{myclaim}
\begin{proof}[of Claim~\ref{correcntessclaim}]
    ``$\Longrightarrow$'': Let $\langle \mathcal{A}, \varphi \rangle$ be a positive instance. 
    Hence, $\mathcal A\models\varphi$. 
    According to Claim~\ref{lemma norm}, the choice of the subsets $A_1,A_2,A_3$ is correct and the sum $t$ is reached.
    Adding NOWAIT yields $t'$ in the correct choice of numbers.
    If $A_2'$ does not correspond to a valid assignment, then according to Claim~\ref{lemma fix} the sum $t$ can be reached and $t'$ can be reached by adding WAIT-numbers. 
    As a result, $f(\langle \mathcal{A}, \varphi \rangle) \in \pALTSSS$.

    ``$\Longleftarrow$'': We use contraposition to prove that direction. 
    Let $\langle \mathcal{A}, \varphi \rangle$ be a negative instance. 
    Then, for all assignments of $x_1,\dots,x_k$ there is an assignment of $x_{k+1},\dots,x_{k+l}$ such that for all assignments of $x_{k+l+1},$ $\dots,x_{k+l+m}$ the quantifier-free part, so at least one atomic formula, is not satisfied. 
    If subsets are chosen that obey Claim~\ref{lemma norm}, the sum $t$ cannot be reached. 
    If one chose subsets  from $A_1$ or $A_3$ that do not correspond to a valid assignment, then according to Claim~\ref{lemma kaputt} the sum $t$ cannot be reached by using only numbers from VAR, ATOM, NORM, and FIX. 
    That is because choosing FIX$(x_j)$ would create a number in the $B_j$-block that is larger than $(l+1) \cdot |A|$ and $t$ cannot be reached anymore.

    The only way to reach $1$ at every position of the $L$-blocks is to choose numbers from ATOM. 
    If a number of the form FIX$(x_j, d)$ is chosen, then the $B_j$-block would have a number that is larger than $(l+1) \cdot |A|$. 
    As a result, no number from ATOM can be used for that $j$.
    Furthermore, we cannot use overflows to achieve this. 
    The FIX-numbers are also insufficient to reach $t$. 
    Finally, as $t$ cannot be reached, $t'$ cannot be reached either, as the WAIT-numbers have no influence on these digits of $t$. 
    Hence, $f(\langle \mathcal{A}, \varphi \rangle) \notin \pALTSSS$.
    \hfill$\blacksquare$
\end{proof}
This results in showing that $\pALTSSS$ is $\A3$-hard.\hfill$\Box$
\end{proof}

\begin{proof}[of Theorem~\ref{sss vst}]
    By Lemmas~\ref{SSS3-hard} and \ref{SSS3-in} the result is proven.\hfill$\Box$
\end{proof}

We state a generalisation of the problem $\ALTSSS$. 
Also, we consider the complements of the problems of this generalisation.
For odd $\ell\in\mathbb N^+$, we define the problem $\ALTTSSS$ as follows:

\[
\left\{ \langle A_1,\dots,A_l,k_1,\dots,k_\ell,t \rangle\;\middle |\; \parbox{7cm}{$A_1,\dots, A_\ell\subseteq\mathbb{N}_0$, $k_1,\dots,k_\ell,t \in \mathbb{N}_0$, $\exists A'_1  \subseteq A_1$ s.t.\ $|A_1'| = k_1$, $\forall A'_2 \subseteq A_2$ s.t.\ $|A'_2| = k_2,\dots,\;\exists A_\ell' \subseteq A_\ell$ s.t.\ $|A_\ell'| = k_\ell$, and $\sum_{a \in A_1'} a + \dots + \sum_{a \in A_\ell'} a = t $ } \right\}
\]
Its parameterised version is $\pALTTSSS$ and is defined as
\[
(\ALTTSSS, \langle A_1,\dots A_l,k_1,\dots,k_\ell,t \rangle \mapsto \sum_{i \in [1,\ell]}k_i)
\]
Note that in the definition of the complement problem of $\ALTTSSS$ the quantification of the sets just flips (i.e., $\forall$ becomes $\exists$ and vice versa).
Analogously, $\pcoALTTSSS$ is defined.

\begin{corollary}
    Let $\ell \in \mathbb{N}^+$. 
    If $\ell$ is odd then $\pALTTSSS$ is $\A \ell$-complete under $\leqfpt$-reductions. 
    If $\ell$ is even then $\pcoALTTSSS$ is $\A \ell$-complete under $\leqfpt$-reductions.
\end{corollary}
\begin{proof} 
    The membership of both problems can be shown analogously to the proof of Lemma~\ref{SSS3-in}. 
    In essence, the number of alternations needs to be taken into account. 
    If $\ell$ is odd, the construction from Lemma~\ref{SSS3-hard} can be adapted. 
    The VAR-numbers just need to be distributed over more sets. 
    If $\ell$ is even, we show that $\PMC(\simple\Sigma_l[2]) \leqfpt \pcoALTTSSS$.
    Let $\varphi$ be from a $\PMC(\simple\Sigma_l[2])$ instance, i.e., for $k = \sum_{i \in[1,k]} k_i$ and $x_{i,j} \in \{x_1,\dots,x_k\}$:
    \begin{align*}
        \varphi &= \exists x_1 \dots \forall x_k \bigvee_{i=1}^n \lambda_i(x_{i,1},x_{i,2})\\
        \Longleftrightarrow \varphi &= \neg \neg \exists x_1 \dots \forall x_k \bigvee_{i=1}^n \lambda_i(x_{i,1},x_{i,2}) \Longleftrightarrow \varphi = \neg \forall x_1 \dots \exists x_k \bigwedge_{i=1}^n \neg \lambda_i(x_{i,1},x_{i,2}) 
    \end{align*}
Now, as in the construction in the proof of Lemma~\ref{SSS3-hard}, the last quantifier is existentially. 
The negation in front of the formula is equivalent to not reaching the sum in $\coALTTSSS$. 
Furthermore, the evaluation of $\neg\lambda_i(x_{i,1},x_{i,2})$ in $\coALTTSSS$ is necessary for the ATOM-numbers.
Otherwise, it is a straightfoward adaption of the proof of Lemma~\ref{SSS3-hard}.\hfill$\square$
\end{proof}

\section{Conclusion}
In this paper, we gave a new and more natural characterisation of the $A$-hierarchy in terms of variants of the well known SUBSET-SUM problem. 
Intuitively, the problem is defined via alternations of existential and universal quantifiers that are used to quantify over subsets of natural numbers that have to correctly sum up to a given target value. 
We showed that the parameterised version (parameter sum of the sizes of the to be chosen subsets) of this problem is complete for every level of the $A$-hierarchy depending on the number of alternations.

Independent of its use in this paper, the SUBSET-SUM generalisation might be helpful in the context of counting problems.

\begin{credits}
\subsubsection{\ackname} The second author was partially supported by a grant of the German Research Foundation (DFG) under the project number 511769688 (ME 4279/3-1).

\subsubsection{\discintname}
The authors have no competing interests to declare that are
relevant to the content of this article.
\end{credits}

\bibliographystyle{splncs04}
\bibliography{refs.bib}

\end{document}